\documentclass{article}%
\usepackage{amsmath}
\usepackage{amsfonts}
\usepackage{amssymb}
\usepackage{graphicx}%
\providecommand{\U}[1]{\protect\rule{.1in}{.1in}}
\newtheorem{theorem}{Theorem}

\newtheorem{corollary}[theorem]{Corollary}

\newtheorem{definition}[theorem]{Definition}

\newtheorem{proposition}[theorem]{Proposition}
\newtheorem{remark}[theorem]{Remark}

\newenvironment{proof}[1][Proof]{\noindent\textbf{#1.} }{\ \rule{0.5em}{0.5em}}
\begin{document}

\title{The Relativistic Hamilton-Jacobi Equation for a Massive, Charged and Spinning
Particle, its Equivalent Dirac Equation and the\\de Broglie-Bohm Theory}
\author{Waldyr A. Rodrigues Jr. and Samuel A. Wainer\\Institute of Mathematics, Statistics and Scientific Computation\\IMECC-UNICAMP\\walrod@ime.unicamp.br~~~~~~samuelwainer@ime.unicamp.br}
\maketitle

\begin{abstract}
Using Clifford and Spin-Clifford formalisms we prove that the classical
relativistic Hamilton Jacobi equation for a charged massive (and spinning)
particle interacting with an external electromagnetic field is equivalent to
Dirac-Hestenes equation satisfied by a class of spinor fields that we call
classical spinor fields. These spinor fields are characterized by having the
Takabayashi angle function constant (equal to $0$ or $\pi$). We also
investigate a nonlinear Dirac-Hestenes like equation that comes from a class
of generalized classical spinor fields. Finally, we show that a general
Dirac-Hestenes equation (which is a representative in the Clifford bundle of
the usual Dirac equation) gives a generalized Hamilton-Jacobi equation where
the quantum potential satisfies a severe constraint and the \textquotedblleft
mass of the particle\textquotedblright\ becomes a variable. Our results can
then eventually explain experimental discrepancies found between prediction
for the de Broglie-Bohm theory and recent experiments. We briefly discuss de
Broglie's double solution theory in view of our results showing that it can be
realized, at least in the case of spinning free particles.The paper contains
several Appendices where notation and proofs of some results of the text are presented.

\end{abstract}

\section{Introduction}

In this paper we prove that the relativistic Hamilton-Jacobi equation for a
massive particle charged particle\ moving in Minkowski spacetime and
interacting with an electromagnetic field is equivalent to Dirac equation
satisfied by a special class of Dirac spinor fields\footnote{This special
class of spinor fields will be called classical spinor fields.} (characterized
by having the Takabayashi angle function equal to $0$ or $\pi$). Also, any
Dirac equation satisfied by these classical spinor fields implies in a
corresponding relativistic Hamilton-Jacobi equation.

After proving these result we show that since general spinor fields which are
solutions of the Dirac equation in an external potential have in general
Takabayashi angle functions \cite{boudet,daviau} which are not constant like
$0$ or $\pi$). Thus, the corresponding derived generalized Hamilton-Jacobi
equation (GHJE) besides having a quantum potential which must satisfy a very
severe constraint (see Eq. (\ref{28})) has also a variable mass (which is a
function of the Takabayashi angle function). So, the usual equations derived
from Schr\"{o}dinger equation used in simulations of, e.g., the double slit
experiments, do not take into account that the mass of the particle, which
comes from the generalized HJE becomes a variable. This eventually must
explain discrepancies found between theoretical predictions from the\ de
Broglie and Bohm formalism \cite{debroglie,bh,holland} and some results of
experiments\ and inconsistencies (such as necessity of surreal trajectories)
as related, e.g., in
\cite{ck1,ck2,englertetal,jonesetal,mahleretal,savantetal}. We briefly discuss
also the de Broglie's double solution theory in view of our results, finding
that, at least for the case of a free (spinning particle) it can be realized.

To show the above results we will use Clifford bundle formalism where Dirac
spinor fields are represented by an equivalence class of even sections of the
Clifford bundle $\mathcal{C\ell}(M,\mathtt{\eta})$ of differential forms.
Details about this theory and notation used may be found in \cite{r2004,
mr2004,rc2016} and a resume is given in the Appendix. Here we recall that in
this paper all calculations are done in the Minkowski spacetime $(M\simeq
\mathbb{R}^{4},\boldsymbol{\eta},D,\tau_{\mathbf{\eta}},\uparrow)$.

\section{A Trivial Derivation of the Relativistic Hamilton-Jacobi Equation
(HJE)}

Let $\sigma:\mathbb{R\rightarrow}M$, $s\mapsto\sigma(s)$ be a timelike curve
in spacetime time representing the motion of a particle of mass $m$ and
electrical charge $e$ interacting with an electromagnetic field $F=dA$, where
the potential $A\in\sec%
%TCIMACRO{\tbigwedge \nolimits^{1}}%
%BeginExpansion
{\textstyle\bigwedge\nolimits^{1}}
%EndExpansion
T^{\ast}M\hookrightarrow\sec\mathcal{C\ell}(M,\mathtt{\eta})$ and $F\in\sec%
%TCIMACRO{\tbigwedge \nolimits^{2}}%
%BeginExpansion
{\textstyle\bigwedge\nolimits^{2}}
%EndExpansion
T^{\ast}M\hookrightarrow\sec\mathcal{C\ell}(M,\mathtt{\eta})$. Let
$\sigma_{\ast}$ be the velocity of the particle and define\footnote{Notice
that $\mathbf{\eta}$ is a mapping from $TM\times TM$ to $\mathfrak{F}$, the
set of real functions. So, it defines a mapping $\sec TM\ni\mathbf{X\mapsto
}\ \boldsymbol{\eta}(\mathbf{X},~)=X\in\sec T^{\ast}M$.}%
\begin{equation}
v=\boldsymbol{\eta}(\sigma_{\ast},~)\label{1}%
\end{equation}
\ \noindent as a $1$-form field over $\sigma$. Then, the motion of such
particle, as is well known is governed in classical electrodynamics by the
Lorentz force law, i.e.,%
\begin{equation}
m\dot{v}=ev\lrcorner F.\label{2}%
\end{equation}
where $\dot{v}=dv/ds$ Now, let $V\in\sec%
%TCIMACRO{\tbigwedge \nolimits^{1}}%
%BeginExpansion
{\textstyle\bigwedge\nolimits^{1}}
%EndExpansion
T^{\ast}M\hookrightarrow\sec\mathcal{C\ell}(M,\mathtt{\eta})$ be a vector
field such that
\begin{equation}
\left.  V\right\vert _{\sigma}=v,~~~~V^{2}=1.\label{3}%
\end{equation}

As defined in the Appendix, let $\{x^{\mu}\}$ be global coordinates for $M$ in
Einstein-Lorentz-Poincar\'{e} gauge, let $\{\gamma^{\mu}=dx^{\mu}\}$ be a
basis for $%
%TCIMACRO{\tbigwedge \nolimits^{1}}%
%BeginExpansion
{\textstyle\bigwedge\nolimits^{1}}
%EndExpansion
T^{\ast}M$ and moreover, let us consider that\footnote{Thus the $1$-forms
$\gamma^{\mu}$ satisfy $\gamma^{\mu}\gamma^{\nu}+\gamma^{\nu}\gamma^{\mu
}=2\eta^{\mu\nu}.$} $\gamma^{\mu}\in\sec%
%TCIMACRO{\tbigwedge \nolimits^{1}}%
%BeginExpansion
{\textstyle\bigwedge\nolimits^{1}}
%EndExpansion
T^{\ast}M\hookrightarrow\sec\mathcal{C\ell}(M,\mathtt{\eta})$. In these
coordinates, the Dirac operator $\boldsymbol{\partial}$ and the representative
of the spin-Clifford operator $\boldsymbol{\partial}^{(s)}$ acting on a
representative $\psi\in\sec\mathcal{C\ell}(M,\mathtt{\eta})$ of Dirac-Hestenes
spinor field $\Psi\in P_{\mathrm{Spin}_{1,3}^{0}}(M,\mathtt{\eta})\times
_{\rho}\mathbb{C}^{4}$ in a given spin frame $\Xi\in\sec P_{\mathrm{Spin}%
_{1,3}^{0}}(M,\mathtt{g})$ are both represented by $\gamma^{\mu}\partial_{\mu
}$, i.e.,
\begin{equation}
\boldsymbol{\partial}=\gamma^{\mu}\partial_{\mu}\text{ }%
,~~~\boldsymbol{\partial}^{(s)}=\gamma^{\mu}\partial_{\mu}. \label{4}%
\end{equation}

Now, we can show the following identity\footnote{Observe that identity given
by Eq. (\ref{5}) is also valid in a general Lorentzian manifold.}%
\begin{equation}
\dot{v}=\left.  V\lrcorner(\boldsymbol{\partial}V)\right\vert _{\sigma
}=\left.  V\lrcorner(\boldsymbol{\partial}\wedge V)\right\vert _{\sigma
}=\left.  V\lrcorner(dV)\right\vert _{\sigma} \label{5}%
\end{equation}
and thus we can write Eq. (\ref{2}) as
\begin{equation}
\left.  V\lrcorner\lbrack d(mV-eA)]\right\vert _{\sigma}=0. \label{6}%
\end{equation}
In what follows we suppose that Eq. (\ref{6}) holds for each integral line of
the vector field $\mathbf{V}=\mathtt{\eta}(V,~)$, i.e.,%
\begin{equation}
V\lrcorner\lbrack d(mV-eA)]=0. \label{6a}%
\end{equation}

A sufficient condition for the validity of Eq. (\ref{6a}) is, of course the
existence of a scalar function $S$ such that%

\begin{equation}
mV-eA=-dS=-\boldsymbol{\partial}S. \label{7}%
\end{equation}
We immediately recognize $\Pi:=-\boldsymbol{\partial}S$ as the canonical
momentum and of course taking into account that $V^{2}=1$ we get from Eq.
(\ref{7}) that
\begin{equation}
(\Pi+eA)^{2}=m^{2} \label{8}%
\end{equation}
which is the relativistic Hamilton Jacobi equation \cite{ll}.

\section{A Classical Dirac-Hestenes Equation}

To proceed, we recall that it is always possible to choose a gauge for the
potential such that the components $\Pi_{\mu}$ of the canonical momentum are
constant. We suppose in what follows that the potentials are already in this
gauge. Moreover, we recall that an invertible representative $\psi$ (in the
Clifford bundle) of Dirac-Hestenes spinor field can be written
as\footnote{Recall that $\gamma^{5}=\tau_{\boldsymbol{g}}\in\sec%
%TCIMACRO{\tbigwedge \nolimits^{4}}%
%BeginExpansion
{\textstyle\bigwedge\nolimits^{4}}
%EndExpansion
T^{\ast}M\hookrightarrow\sec\mathcal{C\ell}(M,\mathtt{g})$ is the volume
element $4$-form field.}%

\begin{equation}
\psi=\rho^{1/2}e^{\frac{\gamma^{5}\beta}{2}}\boldsymbol{R}\in\sec(%
%TCIMACRO{\tbigwedge \nolimits^{0}}%
%BeginExpansion
{\textstyle\bigwedge\nolimits^{0}}
%EndExpansion
T^{\ast}M+%
%TCIMACRO{\tbigwedge \nolimits^{2}}%
%BeginExpansion
{\textstyle\bigwedge\nolimits^{2}}
%EndExpansion
T^{\ast}M+%
%TCIMACRO{\tbigwedge \nolimits^{4}}%
%BeginExpansion
{\textstyle\bigwedge\nolimits^{4}}
%EndExpansion
T^{\ast}M)\hookrightarrow\sec\mathcal{C\ell}(M,\mathtt{\eta}) \label{9}%
\end{equation}
where $\rho,\beta\in\sec%
%TCIMACRO{\tbigwedge \nolimits^{0}}%
%BeginExpansion
{\textstyle\bigwedge\nolimits^{0}}
%EndExpansion
T^{\ast}M$ and for each $x\in M$, $\boldsymbol{R}(x)\in\mathrm{Spin}_{1,3}%
^{0}\simeq\mathrm{Sl}(2,\mathbb{C})$ and $\boldsymbol{RR}^{-1}=\boldsymbol{R}%
^{-1}\boldsymbol{R}=1$ and $\boldsymbol{R=\tilde{R}}$ is called a rotor. Of
course. $\psi^{-1}=\boldsymbol{R}^{-1}\rho^{-1/2}e^{-\frac{\gamma^{5}\beta}%
{2}}$.

Next, we choose $\psi$ such that
\begin{equation}
V=\psi\gamma^{0}\psi^{-1}=e^{\gamma^{5}\beta}\boldsymbol{R}\gamma
^{0}\boldsymbol{R}^{-1}. \label{10}%
\end{equation}

Since $V\in\sec%
%TCIMACRO{\tbigwedge \nolimits^{1}}%
%BeginExpansion
{\textstyle\bigwedge\nolimits^{1}}
%EndExpansion
T^{\ast}M\hookrightarrow\sec\mathcal{C\ell}(M,\mathtt{\eta})$, necessarily the
Takabayashi angle $\beta$ must be $0$ or $\pi$. As we said in the introduction
we will call spinor fields satisfying this condition classical spinor fields.
We now write $\boldsymbol{R}=R(\Pi)$ $e^{S\gamma^{21}}$ where for each $x\in
M$, $R(\Pi)\in\mathrm{Spin}_{1,3}^{0}$ is a rotor depending on $\Pi$ (see Eq.
(\ref{19}) below). We then return to Eq. (\ref{7}) and multiply it on the
right with $\psi$ getting%
\begin{equation}
\Pi\psi=-\boldsymbol{\partial}S\psi=m\psi\gamma^{0}-eA\psi=0. \label{11}%
\end{equation}
We now ask: is it possible to define a differential operator $\mathbf{\hat
{\Pi}}$ acting on sections of the Clifford bundle such that
\begin{equation}
\mathbf{\hat{\Pi}}\psi=\Pi\psi. \label{12}%
\end{equation}

The answer is yes if we define
\begin{equation}
\mathbf{\hat{\Pi}}\psi=\boldsymbol{\partial}\psi\gamma^{21} \label{13}%
\end{equation}
and use the particular classical spinor field
\begin{equation}
\psi=R(\Pi)e^{S\gamma^{21}}. \label{14}%
\end{equation}

Indeed, in this case it is%
\begin{equation}
\boldsymbol{\partial}\psi\gamma^{21}=-\boldsymbol{\partial}S\psi=\Pi\psi.
\label{15}%
\end{equation}

So, the classical spinor field given by Eq. (\ref{14}) satisfies the first
order partial differential equation\footnote{Notice that we are using a
natural system of units where the numerical values of Planck constant $\hbar$
and the speed of light $c$ are equal to one.}
\begin{equation}
\boldsymbol{\partial}\psi\gamma^{21}-m\psi\gamma^{0}+eA\psi=0. \label{16}%
\end{equation}

\begin{remark}
\emph{Eq. (\ref{16})} for a general Dirac-Hestenes spinor field is known as
the Dirac-Hestenes equation \emph{\cite{hestenes2015}} which is a
representative in the Clifford bundle \emph{(}see Appendix C\emph{)} of the
traditional Dirac equation for a covariant Dirac spinor field $\Psi\in\sec
P_{\mathrm{spin}^{0}1,3}(M,\mathtt{\eta})\times_{\rho}\mathbb{C}^{4}$ which
read in a chart for $M$ with coordinates in Einstein-Lorentz-Poincar\'{e}
gauge\emph{\footnote{In Eq. (\ref{17}), the $\boldsymbol{\gamma}^{\mu}$, are
Dirac matrices in standard representation.}}
\begin{equation}
i\boldsymbol{\gamma}^{\mu}(\partial_{\mu}-ieA_{\mu})\Psi+m\Psi=0. \label{17}%
\end{equation}

\end{remark}

So, we have proved that there is a class of classical spinor fields (the ones
satisfying Eq. (\ref{14}) such that the relativistic Hamilton-Jacobi equation
is equivalent to the celebrated Dirac equation. Of course, it is a trivial
exercise to show that starting from Eq. (\ref{16}) for a spinor field
satisfying Eq. (\ref{14}) we get the relativistic Hamilton-Jacobi equation.

So, to complete this section we need to determine the rotor $R(\Pi)$. Note
that we must have
\begin{equation}
\Pi+eA=mV=mR\gamma^{0}R^{-1} \label{18}%
\end{equation}
and then as shown in Appendix C we find%
\begin{equation}
R(\Pi)=\frac{m+(\Pi+eA)\gamma^{0}}{\left[  2\left(  m+\Pi_{0}+eA_{0}\right)
\right]  ^{1/2}}. \label{19}%
\end{equation}

\section{A Classical Non Linear Dirac-Hestenes Equation}

Note that supposing the validity of the classical HJE if instead of taking
$\psi$ as in Eq. (\ref{14}) we write
\begin{equation}
\psi=\varrho^{1/2}R(\Pi)e^{S\gamma^{21}}=\psi_{0}e^{S\gamma^{21}}, \label{20}%
\end{equation}
we get%
\begin{equation}
\boldsymbol{\partial}\psi\gamma^{21}=(\boldsymbol{\partial}\ln\psi_{0}%
)\psi\gamma^{21}-\boldsymbol{\partial}S\psi. \label{21}%
\end{equation}
Thus, substituting this result in Eq. (\ref{11}) we get a \emph{nonlinear
}Dirac-Hestenes equation \cite{rc2016}, namely
\begin{equation}
\boldsymbol{\partial}\psi\gamma^{21}=-m\psi\gamma^{0}+eA\psi
-(\boldsymbol{\partial}\ln\psi_{0})\psi\gamma^{21}. \label{22}%
\end{equation}

\begin{remark}
Notice that in this case $\boldsymbol{\partial}\ln\psi_{0}=\frac{1}%
{2}\boldsymbol{\partial}\ln\rho.$
\end{remark}

\section{The GHJE which Follows from the General Dirac Equation}

In this section we start from the Dirac-Hestenes \ equation satisfied by a
general Dirac-Hestenes spinor field whose representative in the Clifford
bundle in a given spin frame is written as%
\begin{equation}
\boldsymbol{\psi}=\boldsymbol{\rho}^{1/2}R(\Pi)e^{\frac{\beta\gamma^{5}}{2}%
}e^{S\gamma^{21}}=\boldsymbol{\psi}_{0}e^{\frac{\beta\gamma^{5}}{2}}%
e^{S\gamma^{21}}=e^{\frac{\beta\gamma^{5}}{2}}\boldsymbol{\psi}_{0}%
e^{S\gamma^{21}}. \label{23}%
\end{equation}
Thus we have
\begin{equation}
\boldsymbol{\partial\psi}\gamma^{21}=(\boldsymbol{\partial}\ln\boldsymbol{\psi
}_{0})\boldsymbol{\psi}\gamma^{21}-\boldsymbol{\partial}S\boldsymbol{\psi
}-\frac{1}{2}\gamma^{5}\boldsymbol{\partial}(\ln\beta)\boldsymbol{\psi}
\label{24}%
\end{equation}
and the Dirac-Hestenes equation becomes
\begin{equation}
(\boldsymbol{\partial}\ln\boldsymbol{\psi}_{0})\boldsymbol{\psi}\gamma
^{21}-\boldsymbol{\partial}S\boldsymbol{\psi}-\frac{1}{2}\gamma^{5}%
\boldsymbol{\partial}(\ln\beta)\boldsymbol{\psi}-m\boldsymbol{\psi}\gamma
^{0}+eA\boldsymbol{\psi}=0. \label{25}%
\end{equation}
Multiplying Eq. (\ref{25}) on the right by $\boldsymbol{\psi}^{-1}$ and
identifying $\Pi=-\boldsymbol{\partial}S$ we get putting%
\begin{equation}
\boldsymbol{\psi}\gamma^{0}\boldsymbol{\psi}^{-1}=e^{\beta\gamma^{5}}V
\end{equation}
that%
\begin{equation}
\Pi=me^{\beta\gamma^{5}}V+eA+(\boldsymbol{\partial}\ln\boldsymbol{\psi}%
_{0})\boldsymbol{\psi}\gamma^{21}\boldsymbol{\psi}^{-1}+\frac{1}{2}\gamma
^{5}\boldsymbol{\partial}(\ln\beta)\boldsymbol{\psi,} \label{26}%
\end{equation}
which can be written as
\begin{equation}
\Pi=m\cos\beta V+eA+m\sin\beta\gamma^{5}V+(\boldsymbol{\partial}%
\ln\boldsymbol{\psi}_{0})\boldsymbol{\psi}\gamma^{21}\boldsymbol{\psi}%
^{-1}+\frac{1}{2}\gamma^{5}\boldsymbol{\partial}(\ln\beta)\boldsymbol{\psi.}
\label{27}%
\end{equation}

Taking into account that $\Pi=-\boldsymbol{\partial}S$ is a $1$-form field, we
must necessarily have, for consistency, the following constraint for any
solution that implies a genuine classical like equation of motion:%
\begin{equation}
\langle m\sin\beta\gamma^{5}V+(\boldsymbol{\partial}\ln\boldsymbol{\psi}%
_{0})\boldsymbol{\psi}\gamma^{21}\boldsymbol{\psi}^{-1}+\frac{1}{2}\gamma
^{5}\boldsymbol{\partial}(\ln\beta)\boldsymbol{\psi\rangle}_{3}=0 \label{28}%
\end{equation}

If the constraint given by Eq. (\ref{28}) is satisfied then the classical like
equation of motion for the particle is the following \emph{generalized}
Hamilton-Jacobi equation%
\begin{equation}
-\boldsymbol{\partial}S=m\cos\beta V+eA+\langle m\sin\beta\gamma
^{5}V+(\boldsymbol{\partial}\ln\boldsymbol{\psi}_{0})\boldsymbol{\psi}%
\gamma^{21}\boldsymbol{\psi}^{-1}+\frac{1}{2}\gamma^{5}\boldsymbol{\partial
}(\ln\beta)\boldsymbol{\psi\rangle}_{1} \label{29}%
\end{equation}

\begin{remark}
The true \textquotedblleft quantum potential\textquotedblright\ is then
\begin{equation}
Q=\langle(\boldsymbol{\partial}\ln\boldsymbol{\psi}_{0})\boldsymbol{\psi
}\gamma^{21}\boldsymbol{\psi}^{-1}+\frac{1}{2}\gamma^{5}\boldsymbol{\partial
}(\ln\beta)\boldsymbol{\psi\rangle}_{1} \label{30}%
\end{equation}
which differs considerably form the usual Bohm quantum potential. Moreover and
contrary to the usual presentations of the de Broglie-Bohm theory the mass
parameter of the\ particle in the generalized Hamilton-Jacobi equation
\emph{(Eq. (\ref{29}))} is not a constant. Instead, it is%
\begin{equation}
m^{\prime}=m\cos\beta. \label{31}%
\end{equation}

Some results analogous to the ones above but involving classical like
equations of motion instead of the generalized Hamilton Jacobi equation
\emph{(Eq. (\ref{29})) }have been obtained by Hestenes in memorable
papers\emph{ \cite{hestenes1,hestenes2,hestenes3}.}
\end{remark}

\section{Description of the Spin}

In the past sections we associated to a massive and charged particle a
Dirac-Hestenes spinor field satisfying Dirac equation which has been shown to
be equivalent to the relativistic HJE. We next show \cite{rvp,rc2016} how to
describe with the same classical spinor field the intrinsic spin of the
particle. In order to do that it is necessary to have in mind the concepts of
Fermi derivative and the Frenet formalism. For the reader's convenience these
concepts are briefly recalled in Appendix C.\smallskip

As in previous sections, the arena for the motion of particles is Minkowski
spacetime $(M\simeq\mathbb{R}^{4},\boldsymbol{\eta},D,\tau_{\boldsymbol{\eta}%
},\uparrow)$ for which there are global tetrad frames. So, let
$\{\boldsymbol{e}_{\mathbf{a}}\}$ $\in\sec\mathbf{P}_{\mathrm{SO}_{1,3}^{e}%
}(M)$ be one of these global tetrad frames. Let $\{\gamma^{\mathbf{a}%
}\},\gamma^{\mathbf{a}}\in\sec\bigwedge\nolimits^{1}T^{\ast}M\hookrightarrow
\sec\mathcal{C\ell}(M,\mathtt{g})$ be the dual frame of $\{\boldsymbol{e}%
_{\mathbf{a}}\}$. Also,$\ $let $\{\gamma_{\mathbf{a}}\}$ be the reciprocal
frame of $\{\gamma^{\mathbf{a}}\}$, i.e., $\gamma^{\mathbf{a}}\cdot
\gamma_{\mathbf{b}}=\delta_{\mathbf{b}}^{\mathbf{a}}$. Suppose moreover that
the \textit{reference frame}\footnote{In Relativity theory a general reference
frame in a general Lorentzian spacetime $(M\simeq\mathbb{R}^{4},\boldsymbol{g}%
,\boldsymbol{D},\tau_{\boldsymbol{g}},\uparrow)$ is a time vector field
$\boldsymbol{Z}$ pointing to the future and such that $\boldsymbol{g}%
(\boldsymbol{Z,Z})=1$. In \cite{rc2016} the reader may find a classification
of reference frames in Lorentzian spacetimes. For a preliminary classification
of reference frames in Riemann-Cartan spacetimes see \cite{gr}.} defined by
$\boldsymbol{e}_{\mathbf{0}}$ is in free fall, i.e., $D_{\boldsymbol{e}%
_{\mathbf{0}}}\boldsymbol{e}_{\mathbf{0}}=0$ and that the spatial axes along
each one of the integral lines of $\boldsymbol{e}_{\mathbf{0}}$ have been
constructed by Fermi transport of spinning gyroscopes. This is translated\ by
the requirement that $D_{\boldsymbol{e}_{\mathbf{0}}}\boldsymbol{e}%
_{\mathbf{i}}=0$, $\mathbf{i=}1,2,3$ , and we have, equivalently
$D_{\boldsymbol{e}_{0}}\gamma^{\mathbf{a}}=0$. We introduce a spin
coframe\footnote{Details are in \cite{rc2016}.} $\Xi\in P_{\mathrm{Spin}%
_{1,3}^{e}}(M)$ such that $s(\pm\Xi)=\{\gamma^{\mathbf{a}}\}$. Now, let $\Psi$
be the representative of an invertible Dirac-Hestenes spinor field
over\ $\sigma$ (the world line of a spinning particle) in the spin coframe
$\Xi$. Let moreover $\{f_{\mathbf{a}}\}$ be Frenet coframe over $\sigma$ such
that $f_{\mathbf{0}}$ satisfies $\mathtt{g}(f_{\mathbf{0}},~)=\sigma_{\ast}$.
Then, since the general form of a representative of an invertible
Dirac-Hestenes over $\sigma$ is $\Psi=\rho^{\frac{1}{2}}e^{\frac{\beta
\gamma^{5}}{2}}\boldsymbol{R}$ we can write for $\beta=0,$ $\pi$
\begin{equation}
f_{\mathbf{a}}=\Psi\gamma_{\mathbf{a}}\Psi^{-1}. \label{frenet11}%
\end{equation}
Recalling Eq. (\ref{frenet3}) from Appendix D, we obtain using Eq.
(\ref{frenet11}) that
\begin{equation}
D_{\boldsymbol{e}_{\mathbf{0}}}\boldsymbol{R}=\frac{1}{2}\Omega_{D}%
\boldsymbol{R}, \label{frenet12}%
\end{equation}
which may be called a \textit{spinor equation of motion} of a classical
spinning particle.\smallskip

Now, let us show that the spinor equation of motion for a \textit{free}
particle is equivalent to the classical Dirac-Hestenes equation.

We observe that in this case,of course $\Omega_{D}=\Omega_{\mathbf{S}}$.
Moreover, we can trivially redefine the Frenet frame in such a way as to have
$\kappa_{3}=0$. Indeed, this can be done by rotating the original frame with
$U=e^{f_{\mathbf{3}}f_{\mathbf{1}}\frac{\alpha}{2}}$ and choosing
$\alpha=\arctan\left(  -\frac{\kappa_{3}}{\kappa_{1}}\right)  $. So, in what
follows we suppose that this choice has already been made. We are interested
in the case where $\kappa_{2}$ is a real constant. Then, Eq. (\ref{frenet12})
becomes%
\begin{equation}
\boldsymbol{D}_{\boldsymbol{e}_{\mathbf{0}}}\boldsymbol{R}=\frac{1}{2}%
\kappa_{\mathbf{2}}f^{\mathbf{2}}f^{\mathbf{1}}\boldsymbol{R}.
\label{frenet13}%
\end{equation}
The solution of Eq. (\ref{frenet13})\ is%
\begin{equation}
\boldsymbol{R}=\boldsymbol{R}_{0}\exp(\frac{\kappa_{2}}{2}\gamma^{\mathbf{2}%
}\gamma^{\mathbf{1}}t), \label{frenet14}%
\end{equation}
where $\boldsymbol{R}_{0}$ is a constant rotor.\ 

To continue we observe that without any loss of generality we can choose a
global tetrad field such that $\gamma^{\mathbf{a}}=\delta_{\mu}^{a}dx^{\mu}$
(where $\{x^{\mu}\}$ are coordinates in Einstein-Lorentz-Poincar\'{e} gauge).

This choice being made we suppose next existence of a covector field $V\in
\sec\bigwedge\nolimits^{1}T^{\ast}M\hookrightarrow\sec\mathcal{C\ell
}(M,\mathtt{g})$ and a classical Dirac-Hestenes spinor field with
representative $\psi\in\sec\mathcal{C\ell}(M,\mathtt{\eta})$ in the spin
coframe $\Xi$ such that $\left.  V\right\vert _{\sigma}=v$ with%

\begin{equation}
\left.  V\right\vert _{\sigma}=\ \left.  \psi\gamma^{0}\psi^{-1}\right\vert
_{\sigma}=\boldsymbol{R}\gamma^{0}\boldsymbol{R}^{-1}. \label{frenetnew}%
\end{equation}
Then, under all these conditions, writing%
\begin{equation}
\psi=\psi_{0}(p)e^{\gamma^{\mathbf{21}}px} \label{frenet14a}%
\end{equation}
where
\begin{equation}
p=\frac{\kappa_{2}}{2}v \label{frenet14b}%
\end{equation}
and $x=x^{\mu}\gamma_{\mu}$, we can rewrite Eq. (\ref{frenet13}), identifying
$\left.  \psi_{0}(p)\right\vert _{\sigma}=\boldsymbol{R}_{0}$ as:
\begin{equation}
\boldsymbol{D}_{\boldsymbol{e}_{\mathbf{0}}}\boldsymbol{R}=\gamma^{0}%
\cdot\boldsymbol{\partial}\psi=\frac{1}{2}\kappa_{2}v^{0}\psi\gamma^{21}.
\label{frenet15}%
\end{equation}
\ which reduces to Eq. (\ref{frenet13}) in a reference frame $\boldsymbol{e}%
_{0}$ such that $\left.  \boldsymbol{e}_{0}\right\vert _{\sigma}=v$.

Putting $m=-\frac{\kappa_{2}}{2}$ we immediately obtain
\begin{equation}
\boldsymbol{\partial}\psi\gamma^{2}\gamma^{1}-m\psi\gamma^{0}=0.
\label{frenet20}%
\end{equation}
the classical Dirac-Hestenes equation. Note that the signal of $\kappa_{2}$
merely defines the sense of rotation in the $\boldsymbol{e}_{2}\wedge
\boldsymbol{e}_{1}$ plane.

The bilinear invariant $\Omega_{\mathcal{S}}\in\sec\bigwedge\nolimits^{2}%
T^{\ast}M\hookrightarrow\mathcal{C\ell}(M,\mathtt{g})$
\begin{equation}
\Omega_{\mathcal{S}}=k\psi\gamma^{2}\gamma^{1}\tilde{\psi} \label{frenet21}%
\end{equation}
is the s\emph{pin biform field}. Note that for our example $\mathcal{S}%
=\star\Omega_{\mathcal{S}}\llcorner v=k\psi\gamma^{3}\tilde{\psi}\in
\sec\bigwedge\nolimits^{1}T^{\ast}M\hookrightarrow\sec\mathcal{C\ell
}(M,\mathtt{\eta})$ and so it may be called the spin covector field.

The classical spinor equation for an electrical spinning particle interacting
with an external electromagnetic field can be easily obtained using the
principle of minimal coupling. We find that $p$ must be substituted by the
canonical momentum and we arrive at Eq. (\ref{16})

\begin{remark}
The results just obtained show that a natural interpretation suggests itself
for the plane `wave function' $\psi$ in the theory just presented. It
describes a kind of \textquotedblleft probability\textquotedblright\ field in
the sense that it describes a whole set of possible particle trajectories
which are non determined, as it is the case in Hamilton-Jacobi theory, unless
appropriate initial conditions for position and momentum are given
\emph{(}something we know is prohibit by Heisenberg uncertain
principle\emph{)}.Thus, it is, not a physical field in any sense. This last
observation agrees with de Broglie opinion in \emph{\cite{debroglie}} which he
uses to develop his theory of the double solution. We will now show how de
Broglie's idea can be realized in our theory.
\end{remark}

\section{Realization of De Broglie's Idea of the Double Solution}

In this section we will examine only the case of a free particle. We recall
that de Broglie \cite{debroglie} tried hard to constructed a theory where the
Dirac-Hestenes for a free particle equation besides having the $\psi
=R(p)e^{S\gamma^{21}}$\footnote{Of course, de Broglie used the standard matrix
formulation for the Dirac equation since the idea of Dirac-Hestenes spinor
fields where not known when he was investigating his theory the double
solution.} with statistical significance also possesses a \textquotedblleft
singular\textquotedblright\ solution of the form%
\begin{equation}
\digamma(x)=\digamma_{0}(x)e^{S(x)\gamma^{21}} \label{ds0}%
\end{equation}
where $\digamma$ is the representative in the Clifford Bundle\ of a real (not
fictitious)\ \emph{classical }Dirac-Hestenes spinor field describing the
motion of a singularity. De Broglie thought that $\digamma_{0}(x)$ must solve
a nonlinear equation. However, this is not the case. We show now that if
$\digamma$ satisfies the Dirac-Hestenes equation then $\digamma_{0}(x)$
satisfies the massless Dirac-Hestenes equation, i.e.,%
\begin{equation}
\boldsymbol{\partial}\digamma_{0}=0. \label{ds1}%
\end{equation}

Indeed, let us calculate $\boldsymbol{\partial}\digamma(x)$. We have
\begin{equation}
\boldsymbol{\partial}\digamma=(\boldsymbol{\partial}\digamma_{0}%
)e^{S\gamma^{21}}+\boldsymbol{\partial}S\digamma_{0}e^{S\gamma^{21}}%
\gamma^{21} \label{ds2}%
\end{equation}

But%
\begin{equation}
-\boldsymbol{\partial}S=P=mV=mF_{0}\gamma^{0}F_{0}^{-1}. \label{ds2a}%
\end{equation}

Using this result in Eq. (\ref{ds2})
\begin{align}
\boldsymbol{\partial}\digamma &  =(\boldsymbol{\partial}\digamma
_{0})e^{S\gamma^{21}}-mV\digamma_{0}e^{S\gamma^{21}}\gamma^{21}\label{ds4}\\
\boldsymbol{\partial}\digamma\gamma^{21}  &  =(\boldsymbol{\partial}%
\digamma_{0})e^{S\gamma^{21}}\gamma^{21}+m\digamma\gamma^{o}%
\end{align}
and since $\digamma$ satisfies the Dirac-Hestenes equation we have that
$\boldsymbol{\partial}\digamma_{0}=0$

The question that immediately comes to mind is: \smallskip

\emph{May }Eq. (\ref{ds1})\emph{ possess solutions satisfying the constraint
(\ref{ds2a}) describing the motion of a massive }(\emph{and spinning})\emph{
free particle moving with a subluminal velocity }$v$\emph{?\smallskip}

The answer to that question is \emph{yes}. Indeed, introducing the potential
$\mathcal{A}\in\sec%
%TCIMACRO{\tbigwedge \nolimits^{1}}%
%BeginExpansion
{\textstyle\bigwedge\nolimits^{1}}
%EndExpansion
T^{\ast}M\hookrightarrow\sec\mathcal{C\ell}(M,\mathtt{g})$ and defining
\begin{equation}
\digamma_{0}:=\boldsymbol{\partial}\mathcal{A} \label{DS5}%
\end{equation}
we have immediately from Eq. (\ref{ds1}) that%
\begin{equation}
\boldsymbol{\partial}^{2}\mathcal{A}=0. \label{ds5a}%
\end{equation}

Now, we found in \cite{rv,rl} that Eq. (\ref{ds5a}) has \emph{subluminal
soliton like }solutions. Putting $x^{0}=t$, $x^{1}=x$, $x^{2}=y$, $x^{3}=z$
and supposing for simplicity that the wave is moving in the $z$-direction, a
subluminal solution rigidly moving with velocity (1-form) $v$ is
\begin{equation}
\mathcal{A=C}\frac{\sin m\xi}{\xi}\sin(\omega t-kz)\gamma^{1}. \label{DS6}%
\end{equation}
with%

\begin{gather*}
\xi=[(x)^{2}+(y)^{2}+\Gamma^{2}(z-vt)^{2}]^{\frac{1}{2}},\\
\Gamma=(1-v^{2})^{-\frac{1}{2}},\\
\omega^{2}-k^{2}=m^{2},\\
0<v=\frac{d\omega}{dk}<1.
\end{gather*}
and where $\mathcal{C}$ is a constant. Then, the moving soliton like object
realizing de Broglie dream is%

\begin{equation}
\digamma_{0}=\mathcal{C}\boldsymbol{\partial}\left(  \frac{\sin m\xi}{\xi}%
\cos(\omega t-kz)\right)  \gamma^{1}. \label{ds7}%
\end{equation}

Observe that in the rest frame of the soliton $\digamma_{0}\in\sec%
%TCIMACRO{\tbigwedge \nolimits^{2}}%
%BeginExpansion
{\textstyle\bigwedge\nolimits^{2}}
%EndExpansion
T^{\ast}M\hookrightarrow\sec\mathcal{C\ell}(M,\mathtt{g})$ is simply
\begin{equation}
\digamma_{0}=\mathcal{C}m\left(  \frac{\sin m[(x)^{2}+(y)^{2}+(z)^{2}%
]^{\frac{1}{2}}}{[(x)^{2}+(y)^{2}+(z)^{2}]^{\frac{1}{2}}}\cos mt^{\prime
}\right)  \gamma^{\prime01}. \label{ds8}%
\end{equation}

\begin{remark}
Since $\gamma^{\prime01}=u\gamma^{01}u^{-1}$ and $u=e^{\chi\gamma^{30}}$ we
see that $F_{0}\gamma^{0}F_{0}^{-1}\in\sec%
%TCIMACRO{\tbigwedge \nolimits^{1}}%
%BeginExpansion
{\textstyle\bigwedge\nolimits^{1}}
%EndExpansion
T^{\ast}M\hookrightarrow\sec\mathcal{C\ell}(M,\mathtt{g})$ and thus qualify
qualifies $F_{0}\gamma^{0}F_{0}^{-1}$ as a $1$-form velocity field according
to \emph{Eq. (\ref{ds2a}).}
\end{remark}

Some nontrivial problems that need to be investigated are: How does the field
$\digamma_{0}$ behaves when it meets an obstacle, e.g., a double slit
apparatus? How to describe the motion of more than one $\digamma$ like field
moving in Minkowski spacetime? We will return to this problem in another paper.

\section{Conclusions}

We proved that the relativistic Hamilton-Jacobi equation for a massive
particle charged particle (moving in Minkowski spacetime) and interacting with
an electromagnetic field is equivalent to a Dirac-Hestenes
equation\footnote{Which is the representative of the usual Dirac equation in
the Clifford bundle (a result recalled in Appendix A).} satisfied by a special
class of Dirac spinor fields\footnote{These special class of spinor fields
will be called classical spinor fields.} (characterized for having the
Takabayashi angle function equal to $0$ or $\pi$). Also, any Dirac-Hestenes
equation satisfied by these classical spinor fields implies in a corresponding
relativistic Hamilton-Jacobi equation.

After proving these results we recalled that general spinor fields which are
solutions of the Dirac equation in an external potential have in general
Takabayashi angle functions \cite{boudet,daviau} which are not constant
functions). Thus, for these general spinor fields the derived generalized
Hamilton-Jacobi equation (GHJE) besides having a quantum potential which must
satisfy a very severe constraint (see Eq. (\ref{28})) has also a variable mass
(which is a function of the Takabayashi angle function). So, the usual
\ Hamilton-Jacobi like equations derived from Schr\"{o}dinger equation used in
simulations of, e.g., the double slit experiments, do not take into account
that the mass of the particle which comes from the GHJE becomes a variable.
This eventually must explain discrepancies found between theoretical
predictions from the\ de Broglie and Bohm formalism and some results of
experiments\ and inconsistencies (such as necessity of surreal trajectories)
as related, by several authors (cited in the Section 1)\footnote{But on this
issue see also \cite{hc1,ghc2016}.}. We have also shown that de Broglie's
double solution theory\ can be realized at least for the case of a spinning
free particle. However, contrary to de Broglie's suggestion \cite{debroglie}
we found that the physical field $\digamma_{0}$ in Eq. (\ref{ds1}) satisfies
the massless Dirac equation instead of a nonlinear field. This finding
motivates the investigation following questions:

(i) how does the field $\digamma_{0}$ behave when it meets an obstacle, e.g.,
a double slit apparatus?

(ii) how do we describe the motion of more than one $\digamma$ like fields
moving in Minkowski spacetime?

These issues will be investigated in another paper.

The above results have been proved using the powerful Clifford and
spin-Clifford bundles formalism where Dirac spinor fields are represented by
an equivalence class of even sections of the Clifford bundle $\mathcal{C\ell
}(M,\mathtt{\eta})$ of differential forms. The Appendices present this
formalism, the nomenclature and proofs of some results appearing in the main text.

Finally, while preparing this version of the paper we have been informed by
Professor Basil Hiley of his papers \cite{hc1,hc2,hc3}. In particular, the
last two papers contain material related to our paper but with different\ and
conflicting results which we intend to discuss in another publication.

\appendix{}

\section{Preliminaries}

Let $M$ be a four dimensional, real, connected, paracompact and non-compact
manifold. We recall that a Lorentzian manifold is a pair $(M,\boldsymbol{g})$,
where $\boldsymbol{g}\in\sec T_{2}^{0}M$ is a Lorentzian metric of signature
$(1,3)$, i.e., $\forall x\in M,T_{x}M\simeq T_{x}^{\ast}M\simeq\mathbb{R}%
^{1,3}$, where $\mathbb{R}^{1,3}$ is the Minkowski vector space. We define a
Lorentzian spacetime $M$ as pentuple $(M,\boldsymbol{g},\boldsymbol{D}%
,\tau_{\boldsymbol{g}},\uparrow)$, where $(M,\boldsymbol{g},\tau
_{\boldsymbol{g}},\uparrow)$) is an oriented Lorentzian manifold (oriented by
$\tau_{\boldsymbol{g}}$) and time oriented by $\uparrow$, and $\boldsymbol{D}$
is the Levi-Civita connection of $\boldsymbol{g}$. Let $\mathcal{U}\subseteq
M$ be an open set covered by coordinates $\{x^{\mu}\}$. Let $\{e_{\mu
}=\partial_{\mu}\}$ be a coordinate basis of $T\mathcal{U}$ and
$\{\boldsymbol{\vartheta}^{\mu}=dx^{\mu}\}$ the dual basis on $T^{\ast
}\mathcal{U}$, i.e., $\boldsymbol{\vartheta}^{\mu}(\partial_{\nu})=\delta
_{\nu}^{\mu}$. If $\boldsymbol{g}=g_{\mu\nu}\boldsymbol{\vartheta}^{\mu
}\otimes\boldsymbol{\vartheta}^{\nu}$ is the metric on $T\mathcal{U}$ we
denote by $\mathtt{g}=g^{\mu\nu}\boldsymbol{\partial}_{\mu}\otimes
\boldsymbol{\partial}_{\nu}$ the metric of $T^{\ast}\mathcal{U}$, such that
$g^{\mu\rho}g_{\rho\nu}=\delta_{\nu}^{\mu}$. We introduce also
$\{\boldsymbol{\partial}^{\mu}\}$ and $\{\boldsymbol{\vartheta}_{\mu}\}$,
respectively, as the reciprocal bases of $\{e_{\mu}\}$ and
$\{\boldsymbol{\vartheta}_{\mu}\}$, i.e., we have
\begin{equation}
\boldsymbol{g}(\boldsymbol{\partial}_{\nu},\boldsymbol{\partial}^{\mu}%
)=\delta_{\nu}^{\mu},~~~\mathtt{g}(\boldsymbol{\vartheta}^{\mu}%
,\boldsymbol{\vartheta}_{\nu})=\delta_{\nu}^{\mu}. \label{p1}%
\end{equation}

In what follows $\mathbf{P}_{\mathrm{SO}_{1,3}^{e}}(M,\boldsymbol{g})$
($P_{\mathrm{SO}_{1,3}^{e}}(M,\mathtt{g})$) denotes the principal bundle of
oriented Lorentz tetrads (cotetrads).

A \emph{spin structure} for a general $m$-dimensional manifold $M$ ($m=p+q$)
equipped with a metric field $\boldsymbol{g}$ is a principal fiber bundle
$\pi_{s}:P_{\mathrm{Spin}_{p,q}^{e}}(M,\mathtt{g})\rightarrow M$, (called the
Spin Frame Bundle) with a group $\mathrm{Spin}_{p,q}^{e}$ such that there
exists a map%
\begin{equation}
\Lambda:P_{\mathrm{Spin}_{p,q}^{e}}(M,\mathtt{g})\rightarrow P_{\mathrm{SO}%
_{p,q}^{e}}(M,\mathtt{g}),
\end{equation}
satisfying the following conditions:

\begin{definition}
\textbf{(i)} $\pi(\Lambda(p))=\pi_{s}(p),\forall p\in P_{\mathrm{Spin}%
_{p,q}^{e}}(M,\mathtt{g})$, where $\pi$ is the projection map of the bundle
$\pi:P_{\mathrm{SO}_{p,q}^{e}}(M,\mathtt{g})\rightarrow M$.

(ii) $\Lambda(pu)=\Lambda(p)\mathrm{Ad}_{u},\forall p\in P_{\mathrm{Spin}%
_{p,q}^{e}}(M,\mathtt{g})$ and $\mathrm{Ad}:\mathrm{Spin}_{p,q}^{e}%
\rightarrow\mathrm{SO}_{p,q}^{e},\break\mathrm{Ad}_{u}(a)=uau^{-1}.$
\end{definition}

Any section of $P_{\mathrm{Spin}_{p,q}^{e}}(M,\mathtt{g})$ is called a\emph{
spin frame field} \emph{(}or simply a spin frame). We shall use symbol $\Xi
\in\sec P_{\mathrm{Spin}_{p,q}^{e}}(M,\mathtt{g})$ to denote a spin frame.

It can be shown that\footnote{Where $Ad:\mathrm{Spin}_{1,3}^{e}\rightarrow
\mathrm{End}(\mathbb{R}_{1,3})$ is such that $Ad(u)a=uau^{-1}$ and
$\rho:\mathrm{SO}_{1,3}^{e}\rightarrow\mathrm{End(}\mathbb{R}_{1,3}\mathrm{)}$
is the natural action of $\mathrm{SO}_{1,3}^{e}$ on $\mathbb{R}_{1,3}$.}
\cite{rc2016}:%
\begin{equation}
\mathcal{C}\ell(M,\mathtt{\eta})=P_{\mathrm{SO}_{1,3}^{e}}(M,\mathtt{\eta
})\times_{\rho}\mathbb{R}_{1,3}=P_{\mathrm{Spin}_{1,3}^{e}}(M,\mathtt{\eta
})\times_{Ad}\mathbb{R}_{1,3}, \label{eq_cliff}%
\end{equation}
and since\footnote{Given the objets $A$ and $B$, $A$ $\hookrightarrow$ $B$
means as usual that $A$ is embedded in $B$ and moreover, $A\subseteq B$. In
particular, recall that there is a canonical vector space isomorphism between
$\bigwedge\mathbb{R}^{1,3}$ and $\mathbb{R}_{1,3}$, which is written
$\bigwedge\mathbb{R}^{1,3}\hookrightarrow\mathbb{R}_{1,3}$. Details in
\cite{cru,lawmi}.} $\bigwedge TM\hookrightarrow\mathcal{C}\ell(M,\mathtt{\eta
})$, sections of $\mathcal{C}\ell(M,\mathtt{\eta})$ (the Clifford fields) can
be represented as a sum of non homogeneous differential forms. Notice that
$\mathbb{R}_{1,3}$ (the so-called spacetime algebra) is the Clifford algebra
associated with a $4$-dimensional real vector space $(\mathbb{R}^{4})$
equipped with a metric $\mathbf{G}$ of signature $(1,3)$. The pair
$(\mathbb{R}^{4},\mathbf{G})$ is denoted $\mathbb{R}^{1,3}$ and called
\emph{Minkowski vector} \emph{space, }which is not to be confounded with
Minkowski spacetime.

For any parallelizable spacetime structure (as it is the case of Minkowski
spacetime used in the main text), we introduce the global tetrad basis
$\boldsymbol{e}_{\alpha},\alpha=0,1,2,3$ on $TM$ and in $T^{\ast}M$ the
cotetrad basis on $\{\boldsymbol{\gamma}^{\alpha}\}$, which are dual basis. We
introduce the reciprocal basis $\{\boldsymbol{e}^{\alpha}\}$ and
$\{\boldsymbol{\gamma}_{\alpha}\}$ of $\{\boldsymbol{e}_{\alpha}\}$ and
$\{\boldsymbol{\gamma}^{\alpha}\}$ satisfying
\begin{equation}
\boldsymbol{g}(\boldsymbol{e}_{\alpha},\boldsymbol{e}^{\beta})=\delta_{\alpha
}^{\beta},~~~\mathtt{g}(\boldsymbol{\gamma}^{\beta},\boldsymbol{\gamma
}_{\alpha})=\delta_{\alpha}^{\beta}. \label{p2}%
\end{equation}

Moreover, recall that\footnote{Where the matrix with entries $\eta
_{\alpha\beta}$ (or $\eta^{\alpha\beta}$) is the diagonal matrix
$(1,-1,-1,-1)$.}%
\begin{equation}
\boldsymbol{g}=\eta_{\alpha\beta}\boldsymbol{\gamma}^{\alpha}\otimes
\boldsymbol{\gamma}^{\beta}=\eta^{\alpha\beta}\boldsymbol{\gamma}_{\alpha
}\otimes\boldsymbol{\gamma}_{\beta},~~~\mathtt{g}=\eta^{\alpha\beta
}\boldsymbol{e}_{\alpha}\otimes\boldsymbol{e}_{\beta}=\eta_{\alpha\beta
}\boldsymbol{e}^{\alpha}\otimes\boldsymbol{e}^{\beta}. \label{P3}%
\end{equation}

In this work we have that exists a spin structure on\ the 4-dimensional
Lorentzian manifold $(M,\boldsymbol{g})$, since $M$ is parallelizable, i.e.,
$P_{\mathrm{SO}_{1,3}^{e}}(M,\mathtt{g})$ is trivial, because of the following
result due to Geroch \cite{geroch1,geroch2}:

\begin{theorem}
\label{gerochh}For a $4-$dimensional Lorentzian manifold$\ (M,\boldsymbol{g}%
)$, a spin structure exists if and only if $P_{\mathrm{SO}_{1,3}^{e}%
}(M,\mathtt{g})$ is a trivial bundle \emph{\cite{geroch1,geroch2}}
\end{theorem}

The basis $\left.  \boldsymbol{\gamma}^{\alpha}\right\vert _{p}$ of
$(T_{p}M,\boldsymbol{g}_{p})\simeq\mathbb{R}^{1,3},p\in M$, generates the
algebra $\mathcal{C}\ell(T_{p}M,\mathtt{g})\simeq\mathbb{R}_{1,3}$. We have
that \cite{rc2016}%
\[
\mathrm{e}=\frac{1}{2}(1+\boldsymbol{\gamma}^{0})\in\mathbb{R}_{1,3}%
\]
is a primitive idempotent of $\mathbb{R}_{1,3}\simeq\mathbb{H}(2)$ (the so
called spacetime algebra)\footnote{We recall that $\mathcal{C\ell}(T_{x}%
^{\ast}M,\eta)\simeq\mathbb{R}_{1,3}$ the so-called spacetime algebra. Also
the even subalgebra of $\mathbb{R}_{1,3}$ denoted $\mathbb{R}_{1,3}^{0}$ is
isomorphic to te Pauli algebra $\mathbb{R}_{3,0}$, i.e., $\mathbb{R}_{1,3}%
^{0}\simeq\mathbb{R}_{3,0}$. The even subalgebra of the Pauli algebra
$\mathbb{R}_{3,0}^{0}:=\mathbb{R}_{3,0}^{00}$ is the quaternion algebra
$\mathbb{R}_{0,2}$, i.e., $\mathbb{R}_{0,2}\simeq\mathbb{R}_{3,0}^{0}$.
Moreover we have the identifications: $\mathrm{Spin}_{1,3}^{0}\simeq
\mathrm{Sl}(2,\mathbb{C})$, $\mathrm{Spin}_{3,0}\simeq\mathrm{SU}(2)$. For the
Lie algebras of these groups we have $\mathrm{spin}_{1,3}^{0}\simeq
\mathrm{sl}(2,\mathbb{C})$,$\ \mathrm{su}(2)\simeq\mathrm{spin}_{3,0}$. The
important fact to keep in mind for the understanding of some of the
identificastions we done below is that $\mathrm{Spin}_{1,3}^{0},\mathrm{spin}%
_{1,3}^{0}\subset\mathbb{R}_{3,0}\subset\mathbb{R}_{1,3}$ and $\mathrm{Spin}%
_{3,0},\mathrm{spin}_{3,0}\subset\mathbb{R}_{0,2}\subset\mathbb{R}_{1,3}%
^{0}\subset\mathbb{R}_{1,3}$.} and%
\[
\mathrm{f}=\frac{1}{2}(1+\boldsymbol{\gamma}^{0})\frac{1}{2}%
(1+i\boldsymbol{\gamma}^{2}\boldsymbol{\gamma}^{1})\in\mathbb{C\otimes
R}_{1,3}%
\]
is a primitive idempotent of $\mathbb{C\otimes R}_{1,3}$. Now, let
$I=\mathbb{R}_{1,3}\mathrm{e}$ and $I_{\mathbb{C}}=\mathbb{C\otimes R}%
_{1,3}\mathrm{f}$ be respectively the minimal left ideals of $\mathbb{R}%
_{1,3}$ and $\mathbb{C\otimes R}_{1,3}$ generated by $\mathrm{e}$ and
$\mathrm{f}$. Any $\phi\in I$ can be written as%
\[
\phi=\psi\mathrm{e}%
\]
with $\psi\in\mathbb{R}_{1,3}^{0}$. Analogously, any $\mathbf{\phi}\in
I_{\mathbb{C}}$ can be written as%
\[
\psi\mathrm{e}\frac{1}{2}(1+i\boldsymbol{\gamma}^{2}\boldsymbol{\gamma}^{1})
\]
with $\psi\in\mathbb{R}_{1,3}^{0}.$ Recall moreover that $\mathbb{C\otimes
R}_{1,3}\simeq\mathbb{R}_{4,1}\simeq\mathbb{C}(4)$. We can verify that%
\[
\left(
\begin{array}
[c]{cccc}%
1 & 0 & 0 & 0\\
0 & 0 & 0 & 0\\
0 & 0 & 0 & 0\\
0 & 0 & 0 & 0
\end{array}
\right)
\]
is a primitive idempotent of $\mathbb{C}(4)$ which is a matrix representation
of $\mathrm{f}$. In that way, there is a bijection between column spinors,
i.e., elements of $\mathbb{C}^{4}$ and the elements of $I_{\mathbb{C}}$.

Recalling that $\mathrm{Spin}_{1,3}^{e}\hookrightarrow\mathbb{R}_{1,3}^{0}$,
we give:

\begin{definition}
The left \emph{(}respectively right\emph{)} real spin-Clifford bundle of the
spin manifold $M$ is the vector bundle $\mathcal{C}\ell_{\mathrm{Spin}}%
^{l}(M,\mathtt{g})=P_{\mathrm{Spin}_{1,3}^{e}}(M,\mathtt{g})\times
_{l}\mathbb{R}_{1,3}$ \emph{(}respectively $\mathcal{C}\ell_{\mathrm{Spin}%
}^{r}(M,\mathtt{g})=P_{\mathrm{Spin}_{1,3}^{e}}(M,\mathtt{g})\times
_{r}\mathbb{R}_{1,3}$\emph{)} where $l$ is the representation of
$\mathrm{Spin}_{1,3}^{e}$ on $\mathbb{R}_{1,3}$ given by $l(a)x=ax$
\emph{(}respectively, where $r$ is the representation of $\mathrm{Spin}%
_{1,3}^{e}$ on $\mathbb{R}_{1,3}$ given by $r(a)x=xa^{-1}$\emph{)}. Sections
of $\mathcal{C}\ell_{\mathrm{Spin}}^{l}(M,\mathtt{g})$ are called left
spin-Clifford fields \emph{(}respectively right spin-Clifford fields\emph{)}.
\end{definition}

\begin{definition}
Let $\mathbf{e,f}\in\mathcal{C}\ell_{\mathrm{Spin}_{1,3}^{e}}^{l}%
(M,\mathtt{g})$ be a primitive global idempotents \emph{\footnote{We know that
global primitive idempotents exist because $M$ is parallelizable.
\[
\mathbf{e=\mathbf{[(\Xi}_{0},\frac{1}{2}(1+\boldsymbol{\gamma}^{0}%
)\mathbf{)]},f}=\mathbf{[(\Xi}_{0},\frac{1}{2}(1+\boldsymbol{\gamma}^{0}%
)\frac{1}{2}(1+i\boldsymbol{\gamma}^{2}\boldsymbol{\gamma}^{1})\mathbf{)]}%
\]
}}, respectively $\mathbf{e^{r},f}^{r}\in\mathcal{C}\ell_{\mathrm{Spin}%
_{1,3}^{e}}^{r}(M,\mathtt{g})$, and let $I(M,\mathtt{g})$ be the subbundle of
$\mathcal{C}\ell_{\mathrm{Spin}_{1,3}^{e}}^{l}(M,\mathtt{g})$ generated by the
idempotent, that is, if $\mathbf{\Psi}$ is a section of $I(M,\mathtt{g}%
)\subset\mathcal{C}\ell_{\mathrm{Spin}_{1,3}^{e}}^{l}(M,\mathtt{g})$, we have%

\begin{equation}
\mathbf{\Psi e}=\mathbf{\Psi},
\end{equation}

A section $\mathbf{\Psi}$ of $I(M,\mathtt{g})$ is called a left ideal
algebraic spinor field.
\end{definition}

\begin{definition}
A Dirac-Hestenes spinor field \emph{(DHSF)} associated with $\mathbf{\Psi}$ is
a section\footnote{$\mathcal{C}\ell_{\mathrm{Spin}_{1,3}^{e}}^{0l}%
(M,\mathtt{g})$ denotes the even subbundle of $\mathcal{C}\ell_{\mathrm{Spin}%
_{1,3}^{e}}^{l}(M,\mathtt{g})$} $\varPsi$ of $\mathcal{C}\ell_{\mathrm{Spin}%
_{1,3}^{e}}^{0l}(M,\mathtt{g})\subset\mathcal{C}\ell_{\mathrm{Spin}_{1,3}^{e}%
}^{l}(M,\mathtt{g})$ such that\footnote{For any $\mathbf{\Psi}$ the DHSF
always exist, see \cite{rc2016}.}
\begin{equation}
\mathbf{\Psi}=\varPsi\mathbf{e}. \label{DHSF}%
\end{equation}

\end{definition}

\begin{definition}
We denote the complexified left spin-Clifford bundle by%
\[
\mathbb{C}\ell_{\mathrm{Spin}_{1,3}^{e}}^{l}(M,\mathtt{g})=P_{\mathrm{Spin}%
_{1,3}^{e}}(M,\mathtt{g})\times_{l}\mathbb{C}\otimes\mathbb{R}_{1,3}\equiv
P_{\mathrm{Spin}_{1,3}^{e}}(M,\mathtt{g})\times_{l}\mathbb{R}_{1,4}.
\]

\end{definition}

\begin{definition}
An equivalent definition of a \emph{DHSF} is the following. Let $\mathbf{\Psi
}\in\sec\mathbb{C}\ell_{\mathrm{Spin}_{1,3}^{e}}^{l}(M,\mathtt{g})$ such that%
\[
\mathbf{\Psi f=\Psi.}%
\]
Then a \emph{DHSF} associated with $\mathbf{\Psi}$ is an even section
$\varPsi$ of $\mathcal{C}\ell_{\mathrm{Spin}_{1,3}^{e}}^{0l}(M,\mathtt{g}%
)\subset\mathcal{C}\ell_{\mathrm{Spin}_{1,3}^{e}}^{l}(M,\mathtt{g})$ such that%
\begin{equation}
\mathbf{\Psi}=\varPsi\mathbf{f}.
\end{equation}

\end{definition}

The matrix representations of $\varPsi$ \ and $\mathbf{\Psi}$ in
$\mathbb{C}(4)$ (denoted by the same letter) in the given spin basis
are\footnote{Note that in Eq. (\ref{D20}) the $\psi_{i}$ are functions from
$M$ to $\mathbb{R}$.}
\begin{equation}
\varPsi=\left(
\begin{array}
[c]{cccc}%
\psi_{1} & -\psi_{2}^{\ast} & \psi_{3} & \psi_{4}^{\ast}\\
\psi_{2} & \psi_{1}^{\ast} & \psi_{4} & -\psi_{3}^{\ast}\\
\psi_{3} & \psi_{4}^{\ast} & \psi_{1} & -\psi_{2}^{\ast}\\
\psi_{4} & -\psi_{3}^{\ast} & \psi_{2} & \psi_{1}^{\ast}%
\end{array}
\right)  ,~~~\mathbf{\Psi}=\left(
\begin{array}
[c]{cccc}%
\psi_{1} & 0 & 0 & 0\\
\psi_{2} & 0 & 0 & 0\\
\psi_{3} & 0 & 0 & 0\\
\psi_{4} & 0 & 0 & 0
\end{array}
\right)  \label{D20}%
\end{equation}

\subsection{The Hidden Geometrical Meaning of Spinors}

\textbf{DHSFs} unveil the hidden geometrical meaning of spinors (and spinor
fields). Indeed, consider $v\in\mathbb{R}^{1,3}\hookrightarrow\mathbb{R}%
_{1,3}$ a timelike \emph{covector}\ such that $v^{2}=1.$ The linear mapping,
belonging to $\mathrm{SO}_{1,3}^{e}$%
\begin{equation}
v\mapsto RvR^{-1}=Rv\tilde{R}=w,~R\in\mathrm{Spin}_{1,3}^{e},
\end{equation}
define a new covector $w$ such that $w^{2}=1.$ We can therefore fix a covector
$v$ and obtain all other unit timelike covectors by applying this mapping.
This same procedure can be generalized to obtain any type of timelike covector
starting from a fixed unit covector $v$. We define the linear mapping%
\begin{equation}
v\mapsto\psi v\tilde{\psi}=z \label{rotate}%
\end{equation}
to obtain $z^{2}=\rho^{2}>0$. Since $z$ can be written as $z=\rho Rv\tilde{R}%
$, we need%
\begin{equation}
\psi v\tilde{\psi}=\rho Rv\tilde{R}. \label{ddhsf}%
\end{equation}
If we write \ $\psi=\rho^{\frac{1}{2}}MR$ we need that $Mv\tilde{M}=v$ and the
most general solution is $M=e^{\frac{\tau_{\boldsymbol{g}}\beta}{2}}$, where
$\tau_{\boldsymbol{g}}=\gamma^{0}\gamma^{1}\gamma^{2}\gamma^{3}\in%
%TCIMACRO{\tbigwedge \nolimits^{4}}%
%BeginExpansion
{\textstyle\bigwedge\nolimits^{4}}
%EndExpansion
\mathbb{R}^{1,3}\hookrightarrow\mathbb{R}_{1,3}$ and $\beta\in\mathbb{R}$ is
called the Takabayasi angle \cite{rc2016,vr2016}. Then IT follows that $\psi$
is of the form%
\begin{equation}
\psi=\rho^{\frac{1}{2}}e^{\frac{\tau_{\boldsymbol{g}}\beta}{2}}R.
\label{takabaiasi}%
\end{equation}

Now, Eq. (\ref{takabaiasi}) shows that $\psi\in\mathbb{R}_{1,3}^{0}%
\simeq\mathbb{R}_{3,0}$. Moreover, we have that $\psi\tilde{\psi}\neq0$ since%
\begin{equation}
\psi\tilde{\psi}=\rho e^{\tau_{\boldsymbol{g}}\beta}=(\rho\cos\beta
)+\tau_{\boldsymbol{g}}(\rho\sin\beta).
\end{equation}

A representative of a DHSF $\varPsi$ in the Clifford bundle $\mathcal{C}%
\ell(M,\mathtt{g})$ relative to a spin frame $\mathbf{\Xi}_{u}$ is a section
$\boldsymbol{\psi}_{\mathbf{\Xi}_{u}}=[(\mathbf{\Xi}_{u},\psi_{\mathbf{\Xi
}_{u}})]$ of $\mathcal{C}\ell^{0}(M,\mathtt{g})$ where $\psi_{\mathbf{\Xi}%
_{u}}\in\mathbb{R}_{1,3}^{0}\simeq\mathbb{R}_{3,0}$. So a DHSF such
$\boldsymbol{\psi}_{\mathbf{\Xi}_{u}}\boldsymbol{\tilde{\psi}}_{\mathbf{\Xi
}_{u}}\neq0$ induces a linear mapping induced by Eq. (\ref{rotate}), which
\emph{rotates} a covector field and \emph{dilate} it.

\section{Description of the Dirac Equation in the Clifford Bundle}

In the main text we utilized as arena for the motion of particles (and fields)
the Minkowski spacetime structure $(M\simeq\mathbb{R}^{4},\boldsymbol{\eta
},D,\tau_{\mathbf{\eta}})$. Object $\boldsymbol{\eta}\in\sec T_{0}^{2}M$ is
the Minkowski metric field and $D$ is the Levi-Civita connection of
$\boldsymbol{\eta}$. Also, $\tau_{\mathbf{\eta}}\in\sec%
%TCIMACRO{\tbigwedge \nolimits^{4}}%
%BeginExpansion
{\textstyle\bigwedge\nolimits^{4}}
%EndExpansion
T^{\ast}M$ defines an orientation. We denote by $\mathtt{\eta}\in\sec
T_{2}^{0}M$ the metric of the cotangent bundle. It is defined as follows. Let
$\{x^{\mu}\}$ be coordinates for\ $M$ in the Einstein-Lorentz-Poincar\'{e}
gauge \cite{rc2016}. Let $\{\boldsymbol{e}_{\mu}=\partial/\partial x^{\mu}\}$
a basis for $TM$ and $\{\gamma^{\mu}=dx^{\mu}\}$ the corresponding dual basis
for $T^{\ast}M$, i.e., $\gamma^{\mu}(\boldsymbol{e}_{\alpha})=\delta_{\alpha
}^{\mu}$. Then, if $\boldsymbol{\eta}=\eta_{\mu\nu}\gamma^{\mu}\otimes
\gamma^{\nu}$ then $\mathtt{\eta}=\eta^{\mu\nu}\boldsymbol{e}_{\mu}%
\otimes\boldsymbol{e}_{\nu}$, where the matrix with entries $\eta_{\mu\nu}$
and the one with entries $\eta^{\mu\nu}$ are the equal to the diagonal matrix
$\mathrm{diag}(1,-1,-1,-1)$. If $a,b\in\sec%
%TCIMACRO{\tbigwedge \nolimits^{1}}%
%BeginExpansion
{\textstyle\bigwedge\nolimits^{1}}
%EndExpansion
T^{\ast}M$\ we write $a\cdot b=\mathtt{\eta}(a,b)$. We also denote by
$\langle\gamma_{\mu}\rangle$ the reciprocal basis of $\{\gamma^{\mu}=dx^{\mu
}\}$, which satisfies $\gamma^{\mu}\cdot\gamma_{\nu}=\delta_{\nu}^{\mu}$.

We denote the Clifford bundle of differential forms in Minkowski spacetime by
$\mathcal{C\ell}(M,\eta)$ and use notations and conventions in what follows as
in \cite{rc2016} and recall the fundamental relation%
\begin{equation}
\gamma^{\mu}\gamma^{\nu}+\gamma^{\nu}\gamma^{\mu}=2\eta^{\mu\nu}. \label{1a}%
\end{equation}

If $\{\boldsymbol{\gamma}^{\mu},~$ $\mu=0,1,2,3\}$ are the Dirac gamma
matrices in the \emph{standard representation} and $\{\gamma_{\mu}%
,~\mu=0,1,2,3\}$ are as introduced above, we define%
\begin{align}
\sigma_{k}  &  :=\gamma_{k}\gamma_{0}\in\sec%
%TCIMACRO{\tbigwedge \nolimits^{2}}%
%BeginExpansion
{\textstyle\bigwedge\nolimits^{2}}
%EndExpansion
T^{\ast}M\hookrightarrow\sec\mathcal{C\ell}^{0}(M,\eta)\text{, }%
k=1,2,3,\label{2a}\\
\mathbf{i}  &  =\gamma_{5}:=\gamma_{0}\gamma_{1}\gamma_{2}\gamma_{3}\in\sec%
%TCIMACRO{\tbigwedge \nolimits^{4}}%
%BeginExpansion
{\textstyle\bigwedge\nolimits^{4}}
%EndExpansion
T^{\ast}M\hookrightarrow\sec\mathcal{C\ell}(M,\eta),\\
\boldsymbol{\gamma}_{5}  &  :=\boldsymbol{\gamma}_{0}\boldsymbol{\gamma}%
_{1}\boldsymbol{\gamma}_{2}\boldsymbol{\gamma}_{3}\in\mathbb{C(}4\mathbb{)}%
\end{align}

Noting that $M$ is parallelizable, given a global spin frame a \emph{covariant
spinor field} can be taken as a mapping $\mathbf{\Psi}:M\rightarrow
\mathbb{C}^{4}$ \ In standard representation of the gamma matrices where
($i=\sqrt{-1}$, $\boldsymbol{\phi},\boldsymbol{\varsigma}:M\rightarrow
\mathbb{C}^{2}$) $\boldsymbol{\psi}$ \ is given by
\begin{equation}
\boldsymbol{\Psi}=\left(
\begin{array}
[c]{c}%
\boldsymbol{\phi}\\
\boldsymbol{\varsigma}%
\end{array}
\right)  =\left(
\begin{array}
[c]{c}%
\left(
\begin{array}
[c]{c}%
m^{0}+im^{3}\\
-m^{2}+im^{1}%
\end{array}
\right) \\
\left(
\begin{array}
[c]{c}%
n^{0}+in^{3}\\
-n^{2}+in^{1}%
\end{array}
\right)
\end{array}
\right)  , \label{3a}%
\end{equation}
and to $\boldsymbol{\Psi}$ there corresponds the \textbf{DHSF} $\psi\in
\sec\mathcal{C\ell}^{0}(M,\eta)$ given by\footnote{Remember the
identification:%
\[
\mathbb{C}(4)\simeq\mathbb{R}_{4,1}\supseteq\mathbb{R}_{4,1}^{0}%
\simeq\mathbb{R}_{1,3}.
\]
}%
\begin{equation}
\psi=\phi+\varsigma\sigma_{3}=(m^{0}+m^{k}\mathbf{i}\sigma_{k})+(n^{0}%
+n^{k}\mathbf{i}\sigma_{k})\sigma_{3}. \label{4a}%
\end{equation}
We then have the useful formulas in Eq. (\ref{5a}) below that one can use to
immediately translate results of the standard matrix formalism in the language
of the Clifford bundle formalism and vice-versa\footnote{$\tilde{\psi}$ is the
reverse of $\psi$. If $A_{r}\in\sec%
%TCIMACRO{\tbigwedge \nolimits^{r}}%
%BeginExpansion
{\textstyle\bigwedge\nolimits^{r}}
%EndExpansion
T^{\ast}M\hookrightarrow\sec\mathcal{C\ell}(M,\eta)$ then $\tilde{A}%
_{r}=(-1)^{\frac{r}{2}(r-1)}A_{r}$.}
\begin{align}
\boldsymbol{\gamma}_{\mu}\boldsymbol{\Psi}  &  \leftrightarrow\gamma_{\mu}%
\psi\gamma_{0},\nonumber\\
i\boldsymbol{\Psi}  &  \leftrightarrow\psi\gamma_{21}=\psi\mathbf{i}\sigma
_{3},\nonumber\\
i\boldsymbol{\gamma}_{5}\boldsymbol{\Psi}  &  \leftrightarrow\psi\sigma
_{3}=\psi\gamma_{3}\gamma_{0},\nonumber\\
\boldsymbol{\bar{\Psi}}  &  =\boldsymbol{\Psi}^{\dagger}\boldsymbol{\gamma
}^{0}\leftrightarrow\tilde{\psi},\nonumber\\
\boldsymbol{\Psi}^{\dagger}  &  \leftrightarrow\gamma_{0}\tilde{\psi}%
\gamma_{0},\nonumber\\
\boldsymbol{\Psi}^{\ast}  &  \leftrightarrow-\gamma_{2}\psi\gamma_{2}.
\label{5a}%
\end{align}

Using the above dictionary the standard (free) Dirac
equation\footnote{$\partial_{\mu}:=\frac{\partial}{\partial x^{\mu}}$.} for a
Dirac spinor field $\boldsymbol{\psi}:M\rightarrow\mathbb{C}^{4}$
\begin{equation}
i\boldsymbol{\gamma}^{\mu}\partial_{\mu}\boldsymbol{\Psi}-m\boldsymbol{\Psi}=0
\end{equation}
translates immediately in the so-called Dirac-Hestenes equation, i.e.,
\begin{equation}
\boldsymbol{\partial}\psi\gamma_{21}-m\psi\gamma_{0}=0.
\end{equation}

\section{Proof of Eq. (\ref{19})}

From Eqs. (\ref{14}), (\ref{15}) and (\ref{16}) \ we have that%
\begin{equation}
\Pi+eA=mV=mR\gamma^{0}R^{-1}=\boldsymbol{L}\gamma^{0}\boldsymbol{L}^{-1}
\label{B1}%
\end{equation}
This means that the rotor $R:M\ni x\mapsto R(x)\in\mathrm{Spin}_{1,3}^{0}$
must be a boost. We also know that every boost must be the exponential of a
biform \cite{zr,lounesto}.

Now, we observe that for any $x\in M$ if $X(x)$ is such that%
\begin{equation}
X(x)\lrcorner((\Pi+eA)\wedge\gamma^{0})=0
\end{equation}
we will have
\begin{equation}
RXR^{-1}=X \label{B3}%
\end{equation}
and thus we conclude that $R$ must be a biform proportional to $C=\gamma
^{0}\wedge(\Pi+eA)$. Indeed, since $(\Pi+eA)\wedge\gamma^{0}$ anticommutes
with $(\Pi+eA)$ and $\gamma^{0}$ we have that%
\begin{equation}
\frac{1}{m}\Pi+eA=R\gamma^{0}R^{-1}=R^{2}\gamma^{0}=V \label{B4}%
\end{equation}

As we can verify by direct computation the have%

\begin{equation}
R(V)=\frac{1+V\gamma^{0}}{1\left[  2\left(  1+V_{0}\right)  \right]  ^{1/2}}
\label{B5}%
\end{equation}
and thus we can write%

\begin{equation}
R(\Pi)=\frac{m+(\Pi+eA)\gamma^{0}}{\left[  2\left(  m+\Pi_{0}+A_{0}\right)
\right]  ^{1/2}} \label{B6}%
\end{equation}
with
\begin{equation}
\cosh\chi=\frac{1}{m}(\Pi_{0}+A_{0}) \label{B7}%
\end{equation}

Now, without any difficult we can show that
\begin{equation}
R(\Pi)=e^{\left(  \frac{\chi}{2}\frac{(\Pi+eA)\wedge\gamma^{0}}{\left\vert
(\Pi+eA)\wedge\gamma^{0}\right\vert }\right)  } \label{B8}%
\end{equation}

For the case where $A=0$ the canonical momentum is $P$, with $P^{2}=m^{2}$. In
this case \cite{dola} Eq. (\ref{B6}) is
\begin{equation}
R(P)=\frac{m+V\gamma^{0}}{\left[  2m\left(  m+P^{0}\right)  \right]  ^{1/2}}.
\label{B10}%
\end{equation}

\section{Rotation and Fermi Transport in the Clifford Bundle Formalism}

Let $\sigma:\mathbb{R}\supset I\rightarrow M$, $\tau\mapsto\sigma(\tau)$ be a
timelike curve on Minkowski spacetime. Let $\mathbf{Y}$ be a vector field over
$\sigma$. As well known \cite{sw} in order for an observer (following the
curve $\sigma$ to decide when a unitary vector $\mathbf{Y}\in(\sigma_{\ast
\tau})^{\perp}$ has the same spatial direction of the unitary vector
$\mathbf{Y}^{\prime}\in(\sigma_{\ast\tau^{\prime}})^{\perp}\ (\tau^{\prime
}\neq\tau)$, he has to introduce the concept of the \emph{Fermi-Walker
connection}. We have the

\begin{proposition}
There exists one and only one connection $\mathcal{F}$ over $\sigma$, such
that
\begin{equation}
\mathcal{F}_{\mathbf{X}}\mathbf{Y}=[\mathbf{p}(\sigma^{\ast}\boldsymbol{D}%
)_{\mathbf{X}}\mathbf{p}+\mathbf{q}(\sigma^{\ast}\boldsymbol{D})_{\mathbf{X}%
}\mathbf{q}]\mathbf{Y} \label{fermi1}%
\end{equation}
for all vector fields $\mathbf{X}$ on $I$ and for all vector fields
$\mathbf{Y}$ over $\sigma.$
\end{proposition}

In Eq. (\ref{fermi1}) $\sigma^{\ast}\boldsymbol{D}$ is the \textit{induced}
connection over $\sigma$\textit{ }of the Levi-Civita connection
$\boldsymbol{D}$ and $\mathcal{F}$ is called the Fermi-Walker connection over
$\sigma$, and we shall use the notations $\mathcal{F}_{\sigma_{\ast}%
},\mathcal{F}/d\tau$ or $\mathcal{F}_{%
%TCIMACRO{\TeXButton{epsilon}{\mbox{\boldmath{$\varepsilon$}}}}%
%BeginExpansion
\mbox{\boldmath{$\varepsilon$}}%
%EndExpansion
_{0}}$ (see below) when convenient. We also will write (by abuse of language)
only $\boldsymbol{D}$, as usual, for $\sigma^{\ast}\boldsymbol{D}$ in what
follows. A proof of the theorem can be found, e.g., in \cite{rc2016}\newline

We recall that a \emph{moving (orthonormal) frame} $\{\epsilon_{\mathbf{a}}\}$
over $\sigma$\ is an orthonormal basis for $T_{\sigma(I)}M$ with
$\epsilon_{\mathbf{0}}=\sigma_{\ast}$. The set $\{\varepsilon^{\mathbf{a}}\},$
$\varepsilon^{\mathbf{a}}\in\sec T_{\sigma(I)}^{\ast}M$ $\hookrightarrow
\sec\mathcal{C\ell}(M,\mathtt{\eta})$ is the dual comoving frame on $\sigma$,
i.e., $\varepsilon^{\mathbf{a}}(\epsilon_{\mathbf{b}})=\delta_{\mathbf{b}%
}^{\mathbf{a}}.$The set $\{\varepsilon_{\mathbf{a}}\}$, $\varepsilon
_{\mathbf{a}}\in\sec T_{\sigma(I)}^{\ast}M$ $\hookrightarrow\sec
\mathcal{C\ell}(M,\mathtt{\eta})$, with $\varepsilon^{\mathbf{a}}%
\cdot\varepsilon_{\mathbf{b}}=\delta_{\mathbf{b}}^{\mathbf{a}}$ is the
reciprocal frame of $\{\varepsilon^{\mathbf{a}}\}$.

Let $\mathbf{X}$ and $\mathbf{Y}$ be vector fields over $\sigma$ and
$Y=\boldsymbol{\eta}(\mathbf{X,}$~$),Y=\boldsymbol{\eta}(\mathbf{Y,}$%
~$)\in\sec T_{\sigma(I)}^{\ast}M$ $\hookrightarrow\mathcal{C\ell
}(M,\mathtt{\eta})$ the physically equivalent $1$-form fields. We have the

\begin{proposition}
\label{fermi proposition}Let be $X,Y$ be form fields over $\sigma$, as defined
above. The Fermi-Walker\emph{%
\index{connection!Fermi-Walker}%
} connection $\mathcal{F}$ satisfies the properties

\emph{(a)} \qquad$\mathcal{F}_{\epsilon_{o}}Y=\boldsymbol{D}_{\epsilon
_{\mathbf{0}}}Y-(\varepsilon_{\mathbf{0}}\cdot Y)a+(a\cdot Y)\varepsilon_{0},$

where \ $a=\boldsymbol{D}_{\epsilon_{\mathbf{0}}}\varepsilon_{0},$ is the
\emph{(}$1$-form\emph{)} acceleration.

\emph{(b) }\qquad${\frac{d}{d\tau}}\ (X\cdot Y)=\mathcal{F}_{\epsilon
_{\mathbf{0}}}X\cdot Y+X\cdot\mathcal{F}_{\epsilon_{\mathbf{0}}}Y$

\emph{(c) }\qquad$\mathcal{F}_{\epsilon_{\mathbf{0}}}\varepsilon_{\mathbf{0}%
}=0$

\emph{(d) }\qquad If \textbf{$X$}$,\mathbf{Y}$\ are vector fields on $\sigma$
such that \textbf{$X$}$_{u},\mathbf{Y}_{u}\in H_{u}\ \forall u\in I$
then\footnote{Recall the notations introduced in Chapter 4, where $\mathtt{g}$
denote the metric in the cotangent bundle.} $\left.  \mathtt{\eta(}%
\mathcal{F}_{\epsilon_{\mathbf{0}}}X,~)\right\vert _{u}$ and $\left.
\mathtt{\ }\mathbf{\eta}\mathtt{(}\mathcal{F}_{\epsilon_{\mathbf{0}}%
}Y,~)\right\vert _{u}\in H_{u}$,$\ \forall u$ and
\begin{equation}
\mathcal{F}_{\epsilon_{\mathbf{0}}}X\cdot Y=\boldsymbol{D}_{\epsilon
_{\mathbf{0}}}X\cdot Y. \label{fermi2}%
\end{equation}

\end{proposition}

For a proof, see, e.g., \cite{rc2016}\smallskip

Now, let $Y_{0}\in\sec T_{\sigma\tau_{0}}^{\ast}M$. Then, by a well known
property of connections there exists one and only one $1$-form field $Y$ over
$\sigma$ such that $\mathcal{F}_{\epsilon_{o}}Y=0$ and $Y(\tau_{0})=Y_{0}$.
So, if $\{\left.  \varepsilon_{\mathbf{a}}\right\vert _{\tau_{0}}\}$ is an
orthonormal basis for $T_{\sigma\tau_{0}}^{\ast}M$ ($\left.  \varepsilon
_{\mathbf{a}}\right\vert _{\tau_{0}}\in T_{\sigma\tau_{0}}^{\ast}M$,
$\mathbf{a}=0,1,2,3$) we have that the $\varepsilon_{\mathbf{a}}$'$s$ such
that $\mathcal{F}_{\epsilon_{\mathbf{o}}}\varepsilon_{\mathbf{a}}=0$ are
orthonormal for any $\tau$, as follows from (b) in proposition
\ref{fermi proposition}.

We then, agree to say that $Y_{1}\in H_{\tau_{1}}^{\ast}$ and $Y_{2}\in
H_{\tau_{2}}^{\ast}$ have the same spatial direction if and only if
$Y_{1}=a^{i}\left.  \varepsilon_{i}\right\vert _{\tau_{1}}$, $Y_{2}%
=a^{i}\left.  \varepsilon_{i}\right\vert _{\tau_{2}}$. This suggests the following

\begin{definition}
\ We say that $Y$ $\in\sec T_{\sigma}^{\ast}M$ is transported without rotation
\emph{(}\textit{Fermi transported}\emph{)} if and only if $\mathcal{F}%
_{\epsilon_{\mathbf{0}}}Y=0$.
\end{definition}

In that case we have
\begin{align}
\boldsymbol{D}_{\epsilon_{\mathbf{0}}}Y=\varepsilon_{\mathbf{0}}%
\cdot{\mbox{\boldmath$\partial$}}Y\equiv\frac{D}{d\tau}Y  &  =(Y\cdot
\varepsilon_{\mathbf{0}})a-(a\cdot Y)\varepsilon_{\mathbf{0}}\nonumber\\
&  =Y\lrcorner(\varepsilon_{\mathbf{0}}\wedge a)=(a\wedge\varepsilon
_{\mathbf{0}})\llcorner Y. \label{fermi7}%
\end{align}

\subsection{Frenet Frames over%
\index{frame!Frenet}
$\sigma$ and the Darboux Biform}

\begin{proposition}
If $\{\varepsilon_{\mathbf{a}}\}$ is a comoving coframe over $\sigma$, then
there exists a unique biform field $\Omega_{\mathrm{D}}$ over $\sigma$, called
the angular velocity\ \emph{(}Darboux biform\emph{)} such that the
$\varepsilon_{\mathbf{a}}$ satisfy the following system of differential
equations
\begin{gather}
\boldsymbol{D}_{\epsilon_{_{\mathbf{0}}}}\varepsilon_{\mathbf{a}}%
=\Omega_{\mathrm{D}}\llcorner\varepsilon_{\mathbf{a}},\nonumber\\
\Omega_{\mathrm{D}}=\frac{1}{2}\omega_{\mathbf{ab}}\varepsilon^{\mathbf{a}%
}\wedge\varepsilon^{\mathbf{b}}=-\frac{1}{2}(\boldsymbol{D}_{\epsilon
_{_{\mathbf{0}}}}\varepsilon_{\mathbf{b}})\wedge\varepsilon^{\mathbf{b}}.
\label{darboux}%
\end{gather}

\end{proposition}

\begin{proof}
The proof is trivial and we also have the following
\end{proof}

\begin{corollary}
If the comoving coframe $\{\varepsilon_{\mathbf{a}}\}$ is Fermi transported
then the angular velocity is $\Omega_{F},$%
\begin{equation}
\Omega_{F}=a\wedge\varepsilon_{\mathbf{0}}. \label{frenet2}%
\end{equation}

\end{corollary}

\subsection{Physical Meaning of Fermi Transport%
\index{Fermi transport!physical meaning}%
}

Suppose that a comoving coframe $\{\varepsilon_{\mathbf{a}}\}$ is Fermi
transported along a timelike curve\ $\sigma$, \ `materialized' by some
particle. The physical meaning associated to Fermi transport is that the
spatial axis of the tetrad $\mathtt{\eta}(\varepsilon_{i},~)=\epsilon
_{i},i=1,2,3$ are to be associated with the orthogonal spatial directions of
three small gyroscopes carried along $\sigma.$

\begin{definition}
A Frenet coframe $\{f_{\mathbf{a}}\}$ over $\sigma$ is a moving coframe over
$\sigma$ such that $f_{\mathbf{0}}=\boldsymbol{g}(\sigma_{\ast}%
,)=\boldsymbol{g}(\epsilon_{\mathbf{0}},)=\varepsilon_{\mathbf{0}}$ and
\begin{align}
\boldsymbol{D}_{\epsilon_{\mathbf{0}}}f_{\mathbf{a}}  &  =\Omega_{\mathrm{D}%
}\llcorner f_{\mathbf{a}},\nonumber\\
\Omega_{\mathrm{D}}  &  =\kappa_{0}f^{\mathbf{1}}\wedge f^{\mathbf{0}}%
+\kappa_{1}f^{\mathbf{2}}\wedge f^{\mathbf{1}}+\kappa_{2}f^{\mathbf{3}}\wedge
f^{\mathbf{2}}, \label{frenet3}%
\end{align}
where $\kappa_{i},i=0,1,2$\textbf{ }is the $i$-curvature, which is the
projection of $\Omega_{\mathrm{D}}$ in the $f^{\mathbf{i+1}}\wedge f^{I}$ plane.
\end{definition}

\begin{definition}
We say that a $1$-form field $Y$ over $\sigma$ is rotating if and only if it
is rotating in relation to gyroscopes axis, i.e., if\ $\mathcal{F}%
_{\epsilon_{\mathbf{0}}}Y\neq0$.
\end{definition}

\subsection{Rotation 2-form, the\ Pauli-Lubanski Spin $1$-Form and Classical
Spinning Particles}

From Eq. (\ref{frenet3}) taking into account that $a=\boldsymbol{D}%
_{\epsilon_{\mathbf{0}}}f_{\mathbf{0}}=\kappa_{\mathbf{0}}f^{\mathbf{1}}$ and
Eq. (\ref{frenet2}) we can write%
\begin{equation}
\Omega_{\mathrm{D}}=a\wedge f_{\mathbf{0}}+\Omega_{\mathbf{S}}.
\label{frenet4}%
\end{equation}

We now show that the 2-form $\Omega_{\mathbf{S}}$ over $\sigma$ is directly
related (a dimensional factor apart) with the spin 2-form of a
\textit{classical spinning particle}. More, we show now that the Hodge dual of
$\Omega_{\mathbf{S}}$ is associated with the Pauli-Lubanski $1$-form. To have
a notation as closely as possible the usual ones of physical textbooks, let us
put $f_{\mathbf{0}}=v$.

Call $\star\Omega_{\mathbf{S}}$ the Hodge dual of $\Omega_{\mathbf{S}}$. It is
the $2$-form over $\sigma$ given by
\begin{equation}
\star\Omega_{\mathbf{S}}=-\Omega_{\mathbf{S}}f^{\mathbf{5}},\text{
}f^{\mathbf{5}}=f^{\mathbf{0}}f^{\mathbf{1}}f^{\mathbf{2}}f^{\mathbf{3}}.
\label{frenet5}%
\end{equation}

Define the \textit{rotation} $1$-form $\mathbf{S}$ over $\sigma$ by%
\begin{equation}
\mathbf{S}=-\star\Omega_{\mathbf{S}}\llcorner v. \label{frenet6}%
\end{equation}

Since $\Omega_{\mathbf{S}}\llcorner v=0$, we have immediately that
\begin{equation}
\mathbf{S}\cdot v=-\star\Omega_{\mathbf{S}}\llcorner(v\wedge v)=0.
\label{frenet7}%
\end{equation}

Now, since $\boldsymbol{D}_{\epsilon_{\mathbf{0}}}(\mathbf{S}\cdot v)=0$ we
have that ($\boldsymbol{D}_{\epsilon_{\mathbf{0}}}\mathbf{S}\mathfrak{)}\cdot
a=-\mathbf{S}\cdot a$ and then
\begin{equation}
\boldsymbol{D}_{\epsilon_{\mathbf{0}}}\mathbf{S}=-\mathbf{S}\lrcorner(a\wedge
v), \label{frenet8}%
\end{equation}
and it follows that
\begin{equation}
\mathcal{F}_{\epsilon_{\mathbf{0}}}\mathbf{S}=0. \label{frenet9}%
\end{equation}

It is intuitively clear that we must associate $\mathbf{S}$ with the spin of a
classical spinning particle which follows the worldline $\sigma$. And, indeed,
we define the \textit{Pauli-Lubanski} spin $1$-form by
\begin{equation}
W=k\hbar\mathbf{S},\text{ }\hbar=1, \label{frenet10}%
\end{equation}
where $k>0$ is a real constant and $\hbar$ is Planck constant, which is equal
to $1$ in the natural system of units used in this paper. We recall that as it
is well-known $\boldsymbol{D}_{\epsilon_{\mathbf{0}}}W=-W\lrcorner(a\wedge v)$
is the equation of motion of the intrinsic spin of a classical spinning
particle which is being accelerated by a force producing no torque .

\end{document}